\title{Discrete nonlinear hyperbolic equations. Classification of
integrable cases}

\author{Vsevolod~E.~Adler\thanks{Landau Institute for Theoretical
Physics,
 1A pr. Ak. Semenova, 142432 Chernogolovka, Russia.
 \hbox{E--Mail}: {\tt adler}@{\tt itp.ac.ru}. Partially supported by the
 DFG Research Unit 565 ``Polyhedral Surfaces'' and by the RFFR grant
\# 04-01-00403}\and
 Alexander~I.~Bobenko\thanks{Institut f\"ur Mathematik, Technische
Universit\"at Berlin,
 Strasse des 17. Juni 136, 10623 Berlin, Germany.
 E--Mail: {\tt bobenko}@{\tt math.tu-berlin.de}.
 Partially supported by the DFG Research Unit 565 ``Polyhedral Surfaces''}\and
 Yuri~B.~Suris\thanks{Zentrum Mathematik, Technische Universit\"at
M\"unchen,
 Boltzmannstrasse 3, 85747 Garching, Germany.
 E--Mail: {\tt suris}@{\tt ma.tum.de}.
 Partially supported by the ESF Scientific Programme
 ``Methods of Integrable Systems, Geometry, Applied Mathematics''
(MISGAM)}}

\date{\today}
\documentclass[a4paper,12pt]{article}
\usepackage{amsmath,amssymb,amsthm,epic,eepicemu}
\usepackage[english]{babel}
\usepackage{hyperref}

\advance\textwidth 24mm  \advance\oddsidemargin -12mm
\advance\textheight 25mm \advance\topmargin -15mm

\def\Integer{\mathbb{Z}}
\def\CP{\mathbb{CP}}
\def\cP{\mathcal{P}}
\def\const{\mathop{\rm const}}
\def\sn{\mathop{\rm sn}\nolimits}
\def\a{\alpha}
\def\b{\beta}
\def\g{\gamma}
\def\d{\delta} \def\D{\Delta}
\def\eps{\varepsilon}
\def\ka{\varkappa}
\def\la{\lambda}

\def\ox#1{\overset{#1}{x}}
\def\A#1{\a^{(#1)}}

\theoremstyle{plain}
  \newtheorem{theorem}{Theorem}
  \newtheorem{lemma}{Lemma}
\theoremstyle{definition}
  \newtheorem{definition}{Definition}
\theoremstyle{remark}
  \newtheorem{example}{Example}

\begin{document}
\maketitle \thispagestyle{empty}

\begin{abstract}
We consider discrete nonlinear hyperbolic equations on
quad-graphs, in particular on $\Integer^2$. The fields are
associated to the vertices and an equation $Q(x_1,x_2,x_3,x_4)=0$
relates four fields at one quad. Integrability of equations is
understood as 3D-consistency. The latter is a possibility to
consistently impose equations of the same type on all the faces of
a three-dimensional cube. This allows to set these equations also
on multidimensional lattices $\Integer^N$. We classify integrable
equations with complex fields $x$, and $Q$ affine-linear with
respect to all arguments. The method is based on analysis of
singular solutions.
\end{abstract}


\section{Introduction}
The idea of consistency (or compatibility) is at the core of the
theory of integrable systems. It appears already in the very
definition of complete integrability of a Hamiltonian flow in the
Liouville--Arnold sense, which says that the flow may be included
into a complete family of commuting (compatible) Hamiltonian flows
\cite{A}. Similarly, it is a characteristic feature of soliton
(integrable) partial differential equations that they appear not
separately but are always organized in hierarchies of commuting
(compatible) systems. The condition of the existence of a number of
commuting systems may be taken as the basis of the {\em symmetry
approach} which is used to single out integrable systems in some
general classes and to classify them \cite{MSY}. Another way of
relating continuous and discrete systems, connected with the idea of
compatibility, is based on the notion of {\em B\"acklund
transformations} and the {\em Bianchi permutability theorem}
\cite{Bi}. The latter developed into one of the fundamental
principles of discrete differential geometry
\cite{Bobenko_Suris_2005}.

So, the consistency of discrete equations takes center stage in the
integrability theater. We say that
\begin{quote}
a $d$--dimensional discrete equation possesses the {\em consistency}
 property, if it may be imposed in a consistent way on all
$d$--dimensional sublattices of a $(d+1)$--dimensional lattice
\end{quote}
(a more precise definition will be formulated below). As the above
remarks show, the idea that this notion is closely related to
integrability, is not new. In the case $d=1$ it was used as a
possible definition of integrability of maps in \cite{V}. For $d=2$
a decisive step was made in \cite{Bobenko_Suris_2002a} and
independently in \cite{N}: it was shown that the integrability in
the usual sense of soliton theory (as existence of a zero curvature
representation) {\em follows} for two-dimensional systems from the
three-dimensional consistency. So the latter property may be
considered as a definition of integrability. It is a criterion which
may be checked in a completely algorithmic manner starting with no
more information than the equation itself. Moreover, in if this
criterion gives a positive result, it delivers also the discrete
zero curvature representation.

Basic building blocks of systems on quad-graphs are quad-equations,
which are equations on quadrilaterals
\begin{equation}\label{basic eq}
Q(x_1,x_2,x_3,x_4)=0,
\end{equation}
where the field variables $x_i\in\CP^1$ assigned to the four
vertices of the quadrilateral as shown in Figure
\ref{Fig:quadrilateral}. On $\Integer^2$ equations of this type
can be treated as discrete analogues of nonlinear hyperbolic
equations. Boundary value problems of Goursat type for such
systems were studied in \cite{AdlerVeselov}.

{\em Assumption.} In this paper we assume that $Q$ is {\em
affine-linear}, i.e. a polynomial of degree one in each argument.
This implies that equation (\ref{basic eq}) can be solved for any
variable and the solution is a rational function of other three
variables.
\begin{figure}[htbp]
\begin{center}
\setlength{\unitlength}{0.06em}
\begin{picture}(200,140)(-50,-20)
  \put( 0,  0){\line(1,0){100}}
  \put( 0,100){\line(1,0){100}}
  \put(  0, 0){\line(0,1){100}}
  \put(100, 0){\line(0,1){100}}
  \put(-10,-13){$x_1$}
  \put(99,-13){$x_2$}
  \put(99,110){$x_3$}
  \put(-10,110){$x_4$}
   \put(50,50){$Q$}
\end{picture}
\caption{A quad-equation $Q(x_1,x_2,x_3,x_4)=0$; the variables $x_i$
are assigned to vertices} \label{Fig:quadrilateral}
\end{center}
\end{figure}
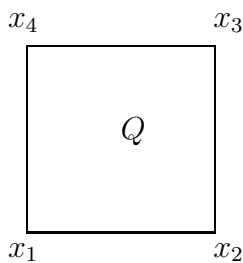

The general idea of integrability as consistency in this case is
shown in Figure \ref{f.cube_0}.
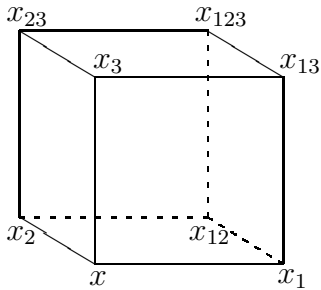
\begin{figure}
\setlength{\unitlength}{0.06em}
\begin{center}
\begin{picture}(100,180)
 \path(0,0)(100,0)(100,100)(0,100)(0,0)(-40,25)(-40,125)(60,125)(100,
100)
 \path(-40,125)(0,100)
 \dashline{3}(-40,25)(60,25)(100,0)\dashline{3}(60,25)(60,125)
 \put(-47,130){$x_{23}$}\put(53,130){$x_{123}$}
 \put(-1,105){$x_3$}\put(98,105){$x_{13}$}
 \put(-47,13){$x_2$}\put(50,13){$x_{12}$}
 \put(-3,-11){$x$}\put(97,-10){$x_1$}
\end{picture}
\end{center}
\caption{A 3D consistent system of quad-equations. The equations are
associated to faces of the cube.} \label{f.cube_0}
\end{figure}
We put six quad-equations on the faces of a coordinate cube. The
subscript $j$ corresponds to the shift in the $j$-th coordinate
direction. If one starts with arbitrary values $x,x_1,x_2,x_3$ then
the values $x_{12},x_{13},x_{23}$ are found from three equations on
the left, front and bottom faces and the equations on the right,
back and top faces yield, in general, three different values of
$x_{123}$. The system is called 3D-consistent, if these three values
are identical for arbitrary initial data $x,x_1,x_2,x_3$.

In \cite{Adler_Bobenko_Suris_2003} we classified 3D-consistent
systems of a particular type. The equations on all faces coincided
up to the values of parameters associated with three types of the
edges. Moreover, cubic symmetry was imposed as well as a certain
additional condition called the tetrahedron property. A
3D-consistent system without the tetrahedron property was found in
\cite{Hietarinta}. Later, this system was shown to be linearizable
in \cite{Ramani_etal}. In \cite{Viallet} it was shown that the
3D-consistent equations classified in
\cite{Adler_Bobenko_Suris_2003} satisfy the integrability test
based on the notion of algebraic entropy.

The consistency approach was generalized in various directions.
Systems with the fields on edges lead to Yang-Baxter maps
\cite{Veselov_Kyoto, SV, PTV}. Quadrirational Yang-Baxter maps were
classified in \cite{Adler_Bobenko_Suris_2004}. The 4D-consistency of
discrete 3D-systems is related to the functional tetrahedron
equation studied in \cite{MN, Ko, KKS, BazhSerg}.

In the present paper we classify 3D-consistent affine-linear
quad-equations in a more general setup. The faces of the consistency
cube can carry a priori different quad-equations. Neither symmetry
nor the tetrahedron property are assumed. This leads to a general
classification of integrable quad-equations.

The outline of our approach is the following.

a) By applying discriminant-like operators to successively
eliminate variables one can descend from an affine-linear
polynomial of four variables, associated to a quadrilateral, to
quadratic polynomials of two variables, associated to its edges,
and finally to quartic polynomials of one variable, associated to
its vertices (Section \ref{s:rhq}).

b) By analysis of singular solutions, we prove that the biquadratic
polynomials which come to an edge of the cube from two adjacent
faces coincide up to a constant factor (see Section \ref{s:sing}).
At this point an additional non-degeneracy assumption is needed. We
assume that all the biquadratic polynomials do not have factors of
the form $x-c$ with constant $c$. (Examples of equations without
this property are presented in Section \ref{s:H}).

c) This allows us to associate to each vertex of the cube a
quartic polynomial in the respective variable; the admissible sets
of polynomials are classified modulo M\"obius transformations,
each variable is transformed independently (Section \ref{s:rrh}).

d) Finally we reverse the procedure and reconstruct the biquadratic
polynomials on the edges of the cube and the affine-linear equations
themselves (Section \ref{s:Q}).


\section{Affine-linear and biquadratic polynomials}\label{s:rhq}

Our approach is based on the descent from the faces to the edges and
from the edges to the vertices of the cube. In this Section we
consider a single face and describe this descent irrespective of
3D-consistency. Let $\cP^m_n$ denote the set of polynomials in $n$
variables which are of degree $m$ in each variable. We consider the
following action of M\"obius transformations on polynomials $f\in
\cP^m_n$:
\[
 M[f](x_1,\dots,x_n)=(c_1x_1+d_1)^m\dots(c_nx_n+d_n)^m
 f\Bigl(\frac{a_1x_1+b_1}{c_1x_1+d_1},\dots,\frac{a_nx_n+b_n}{c_nx_n+d_n}\Bigr),
\]
where $a_id_i-b_ic_i=\D_i\ne0$. The operations
\[
 \cP^1_4\stackrel{\d_{x_i,x_j}}{\longrightarrow}
 \cP^2_2\stackrel{\d_{x_k}}{\longrightarrow} \cP^4_1,
 \qquad \d_{x,y}(Q)=Q_xQ_y-QQ_{xy},\quad \d_x(h)=h^2_x-2hh_{xx}
\]
are covariant with respect to M\"obius transformations. (The
subscripts denote partial differentiation). More precisely: if
$Q\in\cP_4^1$, $h\in\cP_2^2$, then
\begin{equation}\label{dMob}
 \d_{x_i,x_j}(M[Q])=\D_i\D_jM[\d_{x_i,x_j}(Q)],\quad
 \d_{x_i}(M[h])=\D_i^2M[\d_{x_i}(h)].
\end{equation}

Further on we will make an extensive use of {\em relative
invariants} of polynomials under M\"obius transformations. For
quartic polynomials $r\in\cP^4_1$ these relative invariants are
well known and can be defined as the coefficients of the
Weierstrass normal form $r=4x^3-g_2x-g_3$. For a given polynomial
$r(x)=r_4x^4+r_3x^3+r_2x^2+r_1x+r_0$ they are given by (see e.g.
\cite{WW})
\begin{align*}
 g_2(r,x)&= \frac1{48}(2rr^{IV}-2r'r'''+(r'')^2)
     = \frac1{12}(12r_0r_4-3r_1r_3+r^2_2),\\
 g_3(r,x)&= \frac1{3456}(12rr''r^{IV}-9(r')^2r^{IV}-6r(r''')^2+6r'r''r'''-
            2(r'')^3)\\
    &= \frac1{432}(72r_0r_2r_4-27r^2_1r_4+9r_1r_2r_3-27r_0r2_3-2r^3_2).
\end{align*}
Under the M\"obius change of $x=x_1$ these quantities are just
multiplied by the constant factors:
\[
 g_k(M[r],x)=\D_1^{2k}g_k(r,x),\quad k=2,3.
\]
For biquadratic polynomials $h\in \cP^2_2$,
\begin{equation}\label{h}
 h(x,y)=h_{22}x^2y^2+h_{21}x^2y+h_{20}x^2
       +h_{12}xy^2+h_{11}xy+h_{10}x+h_{02}y^2+h_{01}y+h_{00},
\end{equation}
the relative invariants are
\begin{align*}
 i_2(h,x,y)&= 2hh_{xxyy}-2h_xh_{xyy}-2h_yh_{xxy}+2h_{xx}h_{yy}+h^2_{xy}=\\
       &= 8h_{00}h_{22}-4h_{01}h_{21}-4h_{10}h_{12}+8h_{02}h_{20}+h^2_{11},
\\
 i_3(h,x,y)&= \frac14\det
  \begin{pmatrix}
    h & h_x & h_{xx}\\
    h_y & h_{xy} & h_{xxy}\\
    h_{yy} & h_{xyy} & h_{xxyy}
  \end{pmatrix}=\det
  \begin{pmatrix}
    h_{22} & h_{21} & h_{20}\\
    h_{12} & h_{11} & h_{10}\\
    h_{02} & h_{01} & h_{00}
  \end{pmatrix}.
\end{align*}
Notice that $i_3$ can be defined as well by the formula
\[
 -4i_3(h,x,y)=\d_{x,y}(\d_{x,y}(h))/h.
\]
Under the M\"obius change of $x=x_1$ and $y=x_2$,
\[
 i_k(M[h],x,y)=\D_1^k\D_2^ki_k(h,x,y),\quad k=2,3.
\]
The following properties of the operations $\d_{x,y}$, $\d_x$ are
proved straightforwardly.

\begin{lemma}\label{l:dd}
For any affine-linear polynomial $Q(x_1,x_2,x_3,x_4)\in\cP^1_4$
there holds:
\begin{gather}
\label{ddQ}
  \d_{x_3}(\d_{x_1,x_2}(Q))=\d_{x_2}(\d_{x_1,x_3}(Q)), \\
\label{idQ}
  i_k(\d_{x_1,x_2}(Q),x_3,x_4)=i_k(\d_{x_3,x_4}(Q),x_1,x_2),\quad k=2,3.
\end{gather}
For any biquadratic polynomial $h(x_1,x_2)\in\cP^2_2$ there holds:
\begin{equation}\label{gdh}
g_k(\d_{x_1}(h),x_2)=g_k(\d_{x_2}(h),x_1),\quad k=2,3.
\end{equation}
\end{lemma}

\def\xl#1{\xleftarrow{\displaystyle#1}}
\def\xr#1{\xrightarrow{\displaystyle#1}}
\def\bu{\bigg\uparrow}
\def\bd{\bigg\downarrow}
Denote $h^{ij}=h^{ji}=\d_{x_k,x_l}(Q)$ where
$\{i,j,k,l\}=\{1,2,3,4\}$. Then Lemma \ref{l:dd} implies the
commutativity of the diagram\medskip

\begin{equation}\label{diagram}
 \begin{array}{ccccc}
 r_4(x_4) & \xl{~\d_{x_3}~} & h^{34}(x_3,x_4) & \xr{~\d_{x_4}~} &
r_3(x_3)\\[1em]
 \d_{x_1}\bu & &~~~~~\bu\d_{x_1,x_2} & &\bu\d_{x_2}\\[1em]
 h^{14}(x_1,x_4) & \xl{\d_{x_2,x_3}} & Q(x_1,x_2,x_3,x_4)
          & \xr{\d_{x_1,x_4}} & h^{23}(x_2,x_3)\\[1em]
 \d_{x_4}\bd & &~~~~~\bd\d_{x_3,x_4} & &\bd\d_{x_3}\\[1em]
 r_1(x_1) & \xl{~\d_{x_2}~} & h^{12}(x_1,x_2)
          & \xr{~\d_{x_1}~} & r_2(x_2)\\[1em]
 \end{array}
\end{equation}
Moreover, biquadratic polynomials on the opposite edges have the
same invariants $i_2,i_3$, and all four quartic polynomials $r_i$
have the same invariants $g_2,g_3$. This diagram can be completed by
the polynomials $h^{13},h^{24}$ corresponding to the diagonals (so
that the graph of the tetrahedron appears), but we will not need
them. The introduced polynomials satisfy also a number of other
interesting relations.

\begin{lemma}\label{l:QQ}
For any affine-linear polynomial $Q(x_1,x_2,x_3,x_4)\in\cP_4^1$
and with the notations
$h^{ij}(x_i,x_j)=\d_{x_k,x_l}(Q)\in\cP_2^2$, the following
identities hold:
\begin{equation}\label{Q14}
 4i_3(h^{12},x_1,x_2)h^{14}=\det\left(\begin{array}{lll}
   h^{12}          & h^{12}_{x_1}       & \ell \\
   h^{12}_{x_2}    & h^{12}_{x_1x_2}    & \ell_{x_2} \\
   h^{12}_{x_2x_2} & h^{12}_{x_1x_2x_2} & \ell_{x_2x_2}
  \end{array}\right),
\end{equation}
where
$\ell=h^{23}_{x_3x_3}h^{34}-h^{23}_{x_3}h^{34}_{x_3}+h^{23}h^{34}_{x_3x_3}$;

\begin{equation}\label{PQ}
  h^{12}h^{34}-h^{14}h^{23}=PQ,\quad
  P=\det\left(\begin{array}{lll}
   Q       & Q_{x_1}    & Q_{x_3} \\
   Q_{x_2} & Q_{x_1x_2} & Q_{x_2x_3} \\
   Q_{x_4} & Q_{x_1x_4} & Q_{x_3x_4}
  \end{array}\right) \in \cP^1_4\,;
\end{equation}
\begin{equation}\label{Qx}
  \frac{2Q_{x_1}}{Q}=\frac{h^{12}_{x_1}h^{34}-h^{14}_{x_1}h^{23}+
  h^{23}h^{34}_{x_3}-h^{23}_{x_3}h^{34}}{h^{12}h^{34}-h^{14}h^{23}}\,.
\end{equation}
\end{lemma}

Identity (\ref{Q14}) shows that $h^{14}$ can be expressed through
three other biquadratic polynomials (provided $i_3(h^{12})\ne0$).
Identity (\ref{PQ}) defines $Q$ as one of the factors in a simple
expression built from $h^{ij}$. Finally, differentiating
(\ref{Qx}) with respect to $x_2$ or $x_4$ leads to a relation of
the form $Q^2=F[h^{12},h^{23},h^{34},h^{14}]$, where $F$ is a
rational expression on $h^{ij}$ and their derivatives. Therefore,
if the biquadratic polynomials on three edges (out of four) is
known, then $Q$ can be found explicitly. Of course it is seen from
Lemma \ref{l:QQ} that not any set of three biquadratic polynomials
comes as $h^{ij}$ from some $Q\in\cP_4^1$. \medskip

Biquadratic polynomials $h^{ij}$ for a given $Q\in\cP_4^1$ are
closely related to singular solutions of the affine-linear
equation
\begin{equation}\label{Q}
 Q(x_1,x_2,x_3,x_4)=0.
\end{equation}
The polynomial $Q\in \cP^1_4$ is assumed irreducible (in
particular, $Q_{x_i}\not\equiv0$: otherwise the polynomial $Q$
should be considered as reducible, since under the change of the
variable $x_i\mapsto 1/x_i$ it turns into $x_iQ$). Obviously,
equation (\ref{Q}) can be solved with respect to any variable: let
$Q=p(x_j,x_k,x_l)x_i+q(x_j,x_k,x_l)$ then $x_i=-q/p$ for the
generic values of $x_j,x_k,x_l$. However, $x_i$ is not determined
if the point $(x_j,x_k,x_l)$ lies on the curve in $(\CP^1)^3$
\begin{equation}\label{pq}
 S_i:\quad p(x_j,x_k,x_l)=q(x_j,x_k,x_l)=0,\quad Q\equiv px_i+q.
\end{equation}
The projection of this curve onto the coordinate plane $(j,k)$ is
exactly the biquadratic $h^{jk}=pq_{x_l}-p_{x_l}q=0$.

\begin{definition}
A solution $(x_1,x_2,x_3,x_4)$ of equation (\ref{Q}) is called {\em
singular} with respect to $x_i\,$, if it satisfies also the equation
$Q_{x_i}(x_1,x_2,x_3,x_4)=0$. The curve $S_i$ is called the {\em
singular curve} for $x_i$.
\end{definition}

\begin{lemma}\label{l:sing}
If a solution $(x_1,x_2,x_3,x_4)$ of equation (\ref{Q}) is singular
with respect to $x_i$, then $h^{jk}=h^{jl}=h^{kl}=0$. Conversely, if
$h^{jk}=0$ for some solution, then this solution is singular with
respect to either $x_i$ or $x_l$.
\end{lemma}
\begin{proof}
Since $h^{jk}=Q_{x_i}Q_{x_l}-QQ_{x_i,x_l}$, the equations $h^{jk}=0$
and $Q_{x_i}Q_{x_l}=0$ are equivalent for the solutions of equation
$Q=0$.
\end{proof}

We will use the following notion of non-degeneracy for biquadratic
polynomials.
\begin{definition}\label{def: nondeg h}
A biquadratic polynomial $h(x,y)\in\cP_2^2$ is called {\em
non-degenerate}, if no polynomial in its equivalence class with
respect to M\"obius transformations is divisible by a factor $x-c$
or $y-c$ (with $c=\const$).
\end{definition}

According to this definition, a non-degenerate polynomial
$h(x,y)\in\cP_2^2$ is either irreducible, or of the form
$(\a_1xy+\b_1x+\g_1y+\d_1)(\a_2xy+\b_2x+\g_2y+\d_2)$ with
$\a_i\d_i\ne\b_i\g_i$. In both cases the equation $h=0$ defines
$y$ as a two-valued function on $x$ and vice versa. On the other
hand, for example the polynomial $h(x,y)=x-y^2$ (considered as
element of $\cP_2^2$) is, according to Definition \ref{def: nondeg
h}, a degenerate biquadratic, since under the inversion $x\mapsto
1/x$ it turns into $x(1-xy^2)$.
\smallskip

A fundamental role in our studies will be played by the following
notion.

\begin{definition}\label{def:Qtype} An affine-linear
function $Q\in\cP_4^1$ is said to be of type $Q$, if all four its
accompanying biquadratics $h^{jk}\in\cP_2^2$ are non-degenerate,
and is said to be of type $H$ otherwise.
\end{definition}


\section{3D-consistency and biquadratic curves}\label{s:sing}

Consider the system of equations
\begin{equation}\label{6Q}
\begin{array}{ll}
 A(x,x_1,x_2,x_{12})=0,\quad & \bar A(x_3,x_{13},x_{23},x_{123})=0, \\
 B(x,x_1,x_3,x_{13})=0,      & \bar B(x_2,x_{12},x_{23},x_{123})=0, \\
 C(x,x_2,x_3,x_{23})=0,      & \bar C(x_1,x_{12},x_{13},x_{123})=0
\end{array}
\end{equation}
on a cube, see Figure \ref{f.cube}. The functions $A,\ldots,\bar
C$ are affine linear (from $\cP_4^1$) and are not a priori
supposed to be related to each other in any way. We will use the
notation $A^{ij}=\d_{x_k,x_l}A$ for the accompanying biquadratic
polynomials.

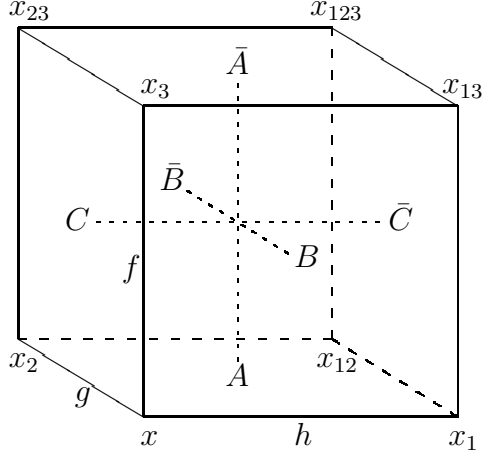
\begin{figure}
\setlength{\unitlength}{0.1em}
\begin{picture}(100,180)(-150,-30)
 \path(0,0)(100,0)(100,100)(0,100)(0,0)(-40,25)(-40,125)(60,125)(100,
100)
 \path(-40,125)(0,100)
 \dashline{3}(-40,25)(60,25)(100,0)\dashline{3}(60,25)(60,125)
 \put(-43,129){$x_{23}$}\put(53,129){$x_{123}$}
 \put(-1,104){$x_3$}\put(95,104){$x_{13}$}
 \put(-43,16){$x_2$}\put(55,16){$x_{12}$}
 \put(-1,-9){$x$}\put(97,-9){$x_1$}
 \put(-7,45){$f$}\put(48,-9){$h$}\put(-22,5){$g$}
 \put(26,10){$A$} \put(26,110){$\bar A$}
 \put(48,48){$B$} \put(5,73){$\bar B$}
 \put(-25,60){$C$} \put(78,60){$\bar C$}
 \dashline{1.5}(46,52.5)(14,72.5)
 \dashline{1.5}(30,18)(30,107)
 \dashline{1.5}(-15,62.5)(75,62.5)
\end{picture}
\caption{A 3D consistent system of quad-equations. The equations are
associated to faces of the cube: $A$ and $\bar A$ on the bottom and
on the top ones, $B$ and $\bar B$ on the front and on the back ones,
$C$ and $\bar C$ on the left and on the right ones, respectively.}
\label{f.cube}
\end{figure}

\begin{theorem}\label{th:h}
Let all six functions $A,\ldots,\bar C$ be of the type $Q$, and let
equations (\ref{6Q}) be 3D-consistent. Then

1) for any edge of the cube, the two biquadratic polynomials
corresponding to this edge (coming from the two faces sharing this
edge) coincide up to a constant factor;

2) the product of these factors around any vertex is equal to $-1$,
for example
\begin{equation}\label{ABCABC}
 A^{0,1}B^{0,3}C^{0,2}+A^{0,2}B^{0,1}C^{0,3}=0;
\end{equation}

3) the system (\ref{6Q}) possesses the tetrahedron property
$\partial x_{123}/\partial x=0$.
\end{theorem}
\begin{proof}
Elimination of $x_{12}$, $x_{13}$ and $x_{23}$ leads to equations
\begin{align*}
 F(\ox2,\ox1_1,\ox1_2,\ox3_3,\ox1_{123}) &=
   \bar A_{x_{13},x_{23}}BC-\bar A_{x_{23}}B_{x_{13}}C
  -\bar A_{x_{13}}BC_{x_{23}}+\bar AB_{x_{13}}C_{x_{23}}=0, \\
 G(\ox2,\ox1_1,\ox3_2,\ox1_3,\ox1_{123}) &=
   \bar B_{x_{12},x_{23}}AC-\bar B_{x_{23}}A_{x_{12}}C
  -\bar B_{x_{12}}AC_{x_{23}}+\bar BA_{x_{12}}C_{x_{23}}=0, \\
 H(\ox2,\ox3_1,\ox1_2,\ox1_3,\ox1_{123}) &=
   \bar C_{x_{12},x_{13}}AB-\bar C_{x_{13}}A_{x_{12}}B
  -\bar C_{x_{12}}AB_{x_{13}}+\bar CA_{x_{12}}B_{x_{13}}=0.
\end{align*}
Here the numbers over the arguments of $F,G,H$ indicate the degrees
of the right hand side in the variables (the degree is understood in
the projective sense, as in the example at the end of the previous
Section). Due to 3D-consistency, the expressions for $x_{123}$ as
functions of $x,x_1,x_2,x_3$ found from these equations, coincide.
Therefore these factorizations hold:
\begin{equation}\label{fgh}
 F=f(x,\ox2_3)K,\quad G=g(x,\ox2_2)K,\quad H=h(x,\ox2_1)K,\quad
 K=K(x,\ox1_1,\ox1_2,\ox1_3,\ox1_{123}),
\end{equation}
where $f,g,h$ are some polynomials of degree 2 in the second
argument. The degree in $x$ remains to be determined.

Let the initial data $x,x_1,x_2$ be free variables, and $x_3$ chosen
to satisfy the equation $f(x,x_3)=0$. Then $F\equiv0$, and thus the
system $B=C=\bar A=0$ does not determine the value of $x_{123}$.
Moreover, the equation $B=0$ can be solved with respect to $x_{13}$
since otherwise the initial data must be constrained by equation
$B^{0,1}(x,x_1)=0$. Analogously, the equation $C=0$ is solvable with
respect to $x_{23}$. Therefore, the uncertainty appears from the
singularity of equation $\bar A=0$ with respect to $x_{123}$. Hence,
the relation $\bar A^{3,13}(x_3,x_{13})=0$ is valid. In virtue of
the assumption of the theorem, $x_{13}$ is a (two-valued) function
of $x_3$ and does  not depend on $x_1$. This means that the equation
$B=0$ is singular with respect to $x_1$ and therefore
$B^{0,3}(x,x_3)=0$. Analogously, $C^{0,3}(x,x_3)=0$.

Thus, we have proven that if $x_3=\varphi(x)$ is a zero of the
polynomial $f$ then it is also a zero of the polynomials
$B^{0,3}$, $C^{0,3}$. If one of these three polynomials is
irreducible, then this already implies that they coincide up to a
constant factor. If the polynomials are reducible this may be
wrong since it is possible that $f=a^2$, $B^{0,3}=ab$,
$C^{0,3}=ac$, where $a,b,c$ are affine-linear in $x$, $x_3$. In
any case we have $\deg_xf=2$ and this is sufficient to complete
the proof.

Indeed, this implies $\deg_xK=0$, that is the tetrahedron property
is valid. In turn, this implies the relation (\ref{ABCABC}), as was
shown in \cite{Adler_Bobenko_Suris_2003}. Recall this calculation:
let us rewrite the system (\ref{6Q}) in the form
\begin{gather*}
 x_{12}=a(x,x_1,x_2),\quad x_{13}=b(x,x_1,x_3),\quad x_{23}=c(x,x_2,
x_3),\\
 x_{123}=d(x_1,x_2,x_3)=\bar a(x_3,x_{13},x_{23})
     =\bar b(x_2,x_{12},x_{23})=\bar c(x_1,x_{12},x_{13})
\end{gather*}
and find by differentiation
\begin{align*}
 d_{x_1}&=\bar a_{x_{13}}b_{x_1},& d_{x_2}&=\bar a_{x_{23}}c_{x_2},&
       0&=\bar a_{x_{13}}b_x+\bar a_{x_{23}}c_x,\\
 d_{x_1}&=\bar b_{x_{12}}a_{x_1},& d_{x_3}&=\bar b_{x_{23}}c_{x_3},&
       0&=\bar b_{x_{12}}a_x+\bar b_{x_{23}}c_x,\\
 d_{x_2}&=\bar c_{x_{12}}a_{x_2},& d_{x_3}&=\bar c_{x_{13}}b_{x_3},&
       0&=\bar c_{x_{12}}a_x+\bar c_{x_{13}}b_x.
\end{align*}
These equations readily imply the relation
\[
 a_{x_2}b_{x_1}c_{x_3}+a_{x_1}b_{x_3}c_{x_2}=0.
\]
The latter is equivalent to (\ref{ABCABC}) by virtue of the
identity $a_{x_2}/a_{x_1}=A^{0,1}/A^{0,2}$. The variables in
equation (\ref{ABCABC}) separate:
$B^{0,3}/C^{0,3}=-A^{0,2}/C^{0,2}\cdot B^{0,1}/A^{0,1}$, so that
$B^{0,3}/C^{0,3}$ may only depend on $x$. In view of the
assumption of the theorem this ratio is constant.
\end{proof}

There exist 3D-consistent systems whose equations are not of type
$Q$. The statements of Theorem \ref{th:h} may or may not be valid
for such a system, as the following examples demonstrate.

\begin{example}
The simplest 3D-consistent equation is the linear one
\[
 x+x_i+x_j+x_{ij}=0.
\]
In this case, all the biquadratic polynomials are equal to $1$, so
that the statement 1) is fulfilled and statement 2) is not. Since 2)
is a consequence of the tetrahedron property 3), the latter cannot
hold either. Indeed,
\[
 x_{123}=2x+x_1+x_2+x_3.
\]
The factor $f$ in this example is also equal to $1$, but this is a
coincidence, destroyed by M\"obius changes of variables. Indeed,
in this case $\deg_xK=1$, and after the inversion of all variables
$x_I\to1/x_I$ we come to $f=xx^2_3$, while $B^{0,3}$ turns into
$x^2x^2_3$.
\end{example}

\begin{example}
The Hietarinta equation \cite{Hietarinta}
\begin{equation}\label{H}
 (x-e^{(j)})(x_{ij}-o^{(j)})(x_i-o^{(i)})(x_j-e^{(i)})
 -(x-e^{(i)})(x_{ij}-o^{(i)})(x_i-e^{(j)})(x_j-o^{(j)})=0
\end{equation}
is 3D-consistent, but the statement 1) does not hold:
\begin{align*}
 B^{0,3}&=(e^{(3)}-o^{(1)})(o^{(1)}-o^{(3)})
          (x-e^{(3)})(x-e^{(1)})(x_3-e^{(1)})(x_3-o^{(3)}),\\
 C^{0,3}&=(e^{(3)}-o^{(2)})(o^{(2)}-o^{(3)})
  (x-e^{(3)})(x-e^{(2)})(x_3-e^{(2)})(x_3-o^{(3)}).
\end{align*}
The factor $f$ is proportional to
$(x-e^{(3)})(x_3-e^{(1)})(x_3-e^{(2)})$. Consequently, $\deg_xK=1$
and the tetrahedron property does not hold.
\end{example}

\begin{example}
Probably the best known example of a 3D-consistent system is given
by the discrete potential KdV equation
\begin{equation}\label{dKdV}
 (x-x_{ij})(x_i-x_j)+\A{i}-\A{j}=0.
\end{equation}
In this case all statements of the theorem are valid, in spite of
the degeneracy of biquadratics:
\[
 B^{0,3}=\A1-\A3,\quad C^{0,3}=\A2-\A3,\quad f=1.
\]
(Recall that the degree is understood in the projective sense.
Under the inversion these polynomials turn into $x^2x^2_3$.)
\end{example}

\begin{example}
Equation (\ref{Q1})
\begin{multline*}
 Q(x,x_1,x_2,x_{12};\A1,\A2;\d)\\
 =\A{1}(x-x_2)(x_1-x_{12})-\A{2}(x-x_1)(x_2-x_{12})+\d\A{1}\A{2}(\A{1}-
\A{2})=0
\end{multline*}
is consistent not only with its own copies (see
\cite{Adler_Bobenko_Suris_2003} and Theorem \ref{th:Qlist} below),
but also with the linear equations. Namely, the system formed of the
equations
\[
 Q(x,x_1,x_{12},x_2;\A1,\A2;\d)=0,\quad
 x_{13}-x_3=x_1-x,\quad x_{23}-x_3=x_2-x
\]
and their copies on the opposite faces, is 3D-consistent. In this
case the edge $(x,x_3)$ carries the polynomials
\[
 B^{0,3}=C^{0,3}=-1,\quad f=1.
\]
However, in contrast to the previous example, the tetrahedron
property is not valid and $\deg_xK=2$. This means that the
polynomial $f$ is not biquadratic and its image under inversion is
$x^2_3$. Moreover, the biquadratic polynomials corresponding to
the edge $(x,x_1)$ do not coincide:
\[
 A^{0,1}=\A2(\A1-\A2)\big((x_1-x)^2-\d(\A1)^2\big),\quad
 B^{0,1}=-1,\quad h=1.
\]
We see in this example that it is possible that some of the
biquadratic polynomials satisfy the assumption of the theorem and
the others do not.
\end{example}

\section{Classification of biquadratic polynomials}\label{s:rrh}

Diagram (\ref{diagram}) suggests an algorithm for the
classification of affine-linear equations $Q=0$ modulo M\"obius
transformations. The first step is to use M\"obius transformations
to bring the polynomials $r_i(x_i)$ associated to the vertices of
the quadrilateral into canonical form. According to formulas
(\ref{dMob}),
\[
 \d_{x_l}(\d_{x_j,x_k}(M[Q]))=\D^2_j\D^2_k\D^2_lM[\d_{x_l}(\d_{x_j,
x_k}(Q))]=\frac{C}{\D^2_i}M[r_i],
\]
where $C=\D_1^2\D_2^2\D_3^2\D_4^2$.  Since the polynomial $Q$ is
defined up to an arbitrary factor, we may assume that M\"obius
changes of variables in the equation $Q=0$ induce transformations
\[
 r_i\mapsto \frac{1}{\D^2_i}M[r_i]
\]
of the polynomials $r_i$. This allows us to bring each $r_i$ into
one of the following six forms:
\[
 r=(x^2-1)(k^2x^2-1),\quad
 r=x^2-1,\quad r=x^2,\quad r=x,\quad r=1,\quad r=0,
\]
according to the six possibilities for the root distribution of
$r$: four simple roots, two simple roots and one double, two pairs
of double roots, one simple root and one triple, one quadruple
root, or, finally, $r$ vanishes identically. Note that in the
first canonical form it is always assumed that $k\neq 0,\pm 1$, so
that the second and third forms are not considered as particular
cases of the first one.

Not every pair of such polynomials is admissible as a pair of
polynomials at two adjacent vertices, since the relative invariants
of the polynomials of such a pair must coincide according to
(\ref{gdh}). We identify all admissible pairs, and then solve the
problem of reconstruction of the biquadratic polynomial (\ref{h}) by
the pair of its discriminants
\begin{equation}\label{hrR}
\begin{aligned}
 \d_y(h)=h^2_y-2hh_{yy}&=r_1(x),\quad
 \d_x(h)=h^2_x-2hh_{xx}&=r_2(y),
\end{aligned}
\end{equation}
which is equivalent to a system of $10$ (bilinear) equations for $9$
unknown coefficients of the polynomial $h$.

\begin{theorem}\label{th:hlist}
Biquadratic polynomials with the pair of discriminants $(r_1(x),
r_2(y))$ in canonical form exist for the following pairs, up to
the permutation of $x$, $y$:
\[
\begin{array}{c|cccccc}
                & (y^2-1)(k^2y^2-1) & y^2-1 & y^2 & y & 1 & 0 \\
 \hline
  (x^2-1)(k^2x^2-1) &      +        &       &     &   &   &   \\
   x^2-1            &               &   +   &  +  &   &   &   \\
   x^2              &               &       &  +  &   &   &   \\
   x                &               &       &     & + & + &   \\
   1                &               &       &     &   & + & + \\
   0                &               &       &     &   &   & +
\end{array}
\]
These polynomials $h$ and their relative invariants $i_2,i_3$ are:
\begin{align}
\nonumber
 & (r(x),r(y)),~ r(x)=(x^2-1)(k^2x^2-1): \\
\label{1111:1111}
 & \qquad h=\frac{1}{2\a}(k^2\a^2x^2y^2+2Axy-x^2-y^2+\a^2),\quad
A^2=r(\a), \\
\nonumber
 & \qquad i_2=3(k^2\a^2+\a^{-2})-k^2-1,\quad 4i_3=A(k^2\a-\a^{-3}); \\
\label{112:112}
 & (x^2-\d,y^2-\d):~
  h=\frac{\a}{1-\a^2}(x^2+y^2)-\frac{1+\a^2}{1-\a^2}xy+\frac{\d(1-
\a^2)}{4\a}\,,\\
\nonumber
 &\qquad i_2=\frac{1+10\a^2+\a^4}{(1-\a^2)^2},\quad
         i_3=\frac{\a^2(1+\a^2)}{(1-\a^2)^3}\,;\\
\label{13:13}
 & (x,y):~ h=\frac{1}{4\a}(x-y)^2-\frac{\a}{2}(x+y)+\frac{\a^3}{4},\quad
  i_2=\frac{3}{4\a^2},~i_3=\frac{1}{32\a^3}\,; \\
\label{22:22}
 & (x^2,y^2):~ h=\la x^2+\mu xy+\nu y^2,\quad
   \mu^2-4\la\nu=1,~i_2=1+12\la\nu,~i_3=-\la\mu\nu; \\
\label{22:22'}
 & \mspace{76mu} h=\la x^2y^2+\mu xy+\nu,\quad
   \mu^2-4\la\nu=1,~i_2=1+12\la\nu,~i_3=\la\mu\nu; \\
\label{4:4}
 & (1,1):~ h=\la(x\pm y)^2+\mu(x\pm y)+\nu,\quad
  \mu^2-4\la\nu=1,~i_2=12\la^2,~i_3=\mp2\la^3; \\
\label{0:0}
 & (0,0):~ h=(\ka xy+\la x+\mu y+\nu)^2,
  \quad i_2=12(\ka\nu-\la\mu)^2,~i_3=2(\ka\nu-\la\mu)^3;\\
\label{112:22}
 & (x^2-1,y^2):~ h=\a y^2\pm xy+\frac{1}{4\a},\quad i_2=1,~i_3=0;\\
\label{13:4}
 & (x,1):~ h=\pm\frac14(y-\a)^2\mp x, \quad i_2=0,~i_3=0; \\
\label{4:0}
 & (1,0):~ h=\la y^2+\mu y+\nu,\quad \mu^2-4\la\nu=1,~i_2=0,~i_3=0.
\end{align}
\end{theorem}
\begin{proof}
The list is obtained by solving the system (\ref{hrR}) for
different canonical pairs $(r_1,r_2)$. The exhaustion of cases is
shortened if we notice that $g^3_2\ne27g^2_3$ in one case only,
and that the relative invariants for the polynomial
$r_1=ax^2+bx+c$ are $12g_2=a^2$, $216g_3=-a^3$, so that the second
polynomial must be of the form $r_2=ay^2+\tilde by+\tilde c$. The
solution for the pair $(x,0)$ turns out to be empty.
\end{proof}


\section{Classification of affine-linear equations of type $Q$}

It is important to note that after bringing the polynomials
$r_i(x_i)$ into canonical forms, one still has some freedom. Namely,
one can use M\"obius transformations which do not change the form of
$r$ to further simplify the biquadratics $h$ and the affine-linear
equation $Q$. In particular, the list of Theorem \ref{th:hlist} is
slightly more detailed than the list of biquadratics modulo M\"obius
transformations.

Indeed, the polynomial (\ref{22:22}) turns into (\ref{22:22'})
under the inversion of $x$; the change $x\mapsto-x$ allows to fix
the signs in the polynomials (\ref{4:4}), (\ref{112:22}); in the
case (\ref{13:4}), the sign is fixed by the change $x\mapsto-x$,
$y\mapsto iy$; the polynomials (\ref{0:0}), (\ref{4:0}) admit a
further simplification.

However, a transformation of any one of the four variables for a
quadrilateral influences biquadratic polynomials on two edges
adjacent to the correspondent vertex, and therefore it cannot a
priori be guaranteed that all four biquadratics can be brought to
some definite form simultaneously. For example, if each vertex
corresponds to the polynomial $r_i=x^2_i$ then the edges may
correspond to the polynomials either of the form (\ref{22:22}) or
of the form (\ref{22:22'}). We do not know a priori that these
polynomials can be always brought into the same form (even with
different coefficients). Actually, this is possible, as the proof
of the following theorem shows.

The next step is the reconstruction of the affine-linear polynomials
from the biquadratic ones. Since our goal is only the classification
of systems of type $Q$ equations, we will not solve this problem in
its full generality. We leave aside the cases (\ref{112:22}),
(\ref{13:4}) and (\ref{4:0}), since the corresponding biquadratics
are degenerate. For the same reason, the additional restrictions on
the values of parameters are imposed: $\la\nu\ne0$ in the cases
(\ref{22:22}), (\ref{22:22'}), $\la\ne0$ in the case (\ref{4:4}),
and $\ka\nu-\la\mu\ne0$ in the case (\ref{0:0}).

\begin{theorem}\label{th:qlist}
Any affine-linear equation of type $Q$ is equivalent, up to
M\"obius transformations, to one of the equations in the following
list:
\begin{align}
\nonumber
 & \sn(\a)\sn(\b)\sn(\a+\b)(k^2x_1x_2x_3x_4+1)
   -\sn(\a)(x_1x_2+x_3x_4) \\
\label{q4}
 &\qquad\qquad -\sn(\b)(x_1x_4+x_2x_3)+\sn(\a+\b)(x_1x_3+x_2x_4)=0,\\
\nonumber
 & (\a-\a^{-1})(x_1x_2+x_3x_4)+(\b-\b^{-1})(x_1x_4+x_2x_3)
   -(\a\b-\a^{-1}\b^{-1})(x_1x_3+x_2x_4)\\
\label{q3}
 &\qquad\qquad +\frac{\d}{4}(\a-\a^{-1})(\b-\b^{-1})(\a\b-\a^{-1}\b^{-
1})=0, \\
\nonumber
 & \a(x_1-x_4)(x_2-x_3)+\b(x_1-x_2)(x_4-x_3) \\
\label{q2}
 & \qquad\qquad
 -\a\b(\a+\b)(x_1+x_2+x_3+x_4)+\a\b(\a+\b)(\a^2+\a\b+\b^2)=0,\\
\label{q1}
 & \a(x_1-x_4)(x_2-x_3)+\b(x_1-x_2)(x_4-x_3)-\d\a\b(\a+\b)=0.
\end{align}
\end{theorem}
\begin{proof}
Let the polynomials $h^{12}$, $h^{23}$, $h^{34}$ and $h^{14}$ be
of the form (\ref{1111:1111}),  with the parameters $(\a,A)$,
$(\b,B)$, $(\tilde\a,\tilde A)$ and $(\tilde\b,\tilde B)$,
respectively, lying on the elliptic curve $A^2=r(\a)$. The
relative invariants $i_2,i_3$ of $h^{12}$ and $h^{34}$ must
coincide as a corollary of (\ref{idQ}), and it is easy to check
that this condition allows only the following possible values for
$(\tilde\a,\tilde A)$:
\[
 (\a,A),\quad (-\a,-A),\quad \frac{1}{k\a^2}(\a,-A),\quad
 \frac{1}{k\a^2}(-\a,A)
\]
and analogously for $(\tilde\b,\tilde B)$. At first glance, it
seems like we had to examine $16$ quadruples of $h^{ij}$, but
actually the situation is much more favorable. Indeed, according
to (\ref{dMob}), a M\"obius change of variables in the equation
$Q=0$ yields
\[
 \delta_{x_k,x_l}(M[Q])=\D_k\D_lM[\delta_{x_k,x_l}(Q)]=
 \frac{C}{\D_i\D_j}M[h^{ij}],
\]
where $C=\D_1\D_2\D_3\D_4$. Since $Q$ is only defined up to a
multiplicative constant, we may assume that a M\"obius change of
variables induces transformations
\[
 h^{ij}\mapsto\frac{1}{\D_i\D_j}M[h^{ij}]
\]
on the biquadratic polynomials $h^{ij}$. In particular, if
\[
h^{34}=h(x_3,x_4,-\alpha,-A)\quad {\rm or}\quad
h^{34}=h\Bigl(x_3,x_4,\frac{1}{k\alpha},-\frac{A}{k\alpha^2}\Bigr),
\]
then the corresponding one of the M\"obius transformations
$x_3\mapsto -x_3$ or $x_3\mapsto 1/(kx_3)$ will change $h^{34}$ to
\[
 -h(-x_3,x_4;-\a,-A),\quad {\rm resp.}\quad
  -kx_3^2h\Bigl(\frac{1}{kx_3},x_4;\frac{1}{k\a},-\frac{A}{k\a^2}\Bigr),
\]
both of which coincide with $h(x_3,x_4,\alpha,A)$ due to
symmetries of the polynomial (\ref{1111:1111}). Thus, performing a
suitable M\"obius transformation of the variable $x_3$ (which does
not affect the polynomial $r(x_3)$), we may assume without loss of
generality that $(\tilde\a,\tilde A)=(\a,A)$. After that, the
polynomial $h^{14}$ is uniquely found according to formula
(\ref{Q14}), and it turns out that the equality $(\tilde\b,\tilde
B)=(\b,B)$ is fulfilled automatically. Thus, the change of one
variable allows to achieve the equality of the parameters
corresponding to the opposite edges of the square. A direct
computation using formula (\ref{Qx}) yields the equation
\[
 \a\b\g(k^2x_1x_2x_3x_4+1)+\a(x_1x_2+x_3x_4)+\b(x_1x_4+x_2x_3)+\g(x_1x_3+x_2x_4)=0,
\]
where $\g=(\a B+\b A)/(k^2\a^2\b^2-1)$, and the final change
$\a\to\sn(\a)$, $A\to\sn'(\a)$ and analogously for $\b$ brings it
to the form (\ref{q4}).

Also in the other cases, suitable M\"obius changes of the
variables $x_2$, $x_3$, $x_4$ allow us to bring the polynomials
into the form $h^{12}=h(x_1,x_2,\a)$, $h^{23}=h(x_2,x_3,\b)$,
$h^{34}=h(x_3,x_4,\a)$. Moreover, a direct computation using
formula (\ref{Q14}) proves that also $h^{14}=h(x_1,x_4,\b)$. Then
the answer is found by use of (\ref{Qx}).

To give a few more details, polynomials (\ref{112:112}) give rise
to equation (\ref{q3}). In this case equations (\ref{idQ}) imply
that the parameters $\a$ of the polynomials $h^{12}$ and $h^{34}$
differ at most by sign. This is compensated by the change
$x_3\to-x_3$ which is possible due to the symmetry
$h(x,y,\a)=-h(-x,y,-\a)$.

In the cases (\ref{22:22}), (\ref{22:22'}), the appropriate scalings
and, if necessary, inversions of the variables $x_2$, $x_3$, $x_4$
allow us to bring $h^{12}$, $h^{23}$, $h^{34}$ into the form
(\ref{112:112}) without the constant term; therefore we simply
obtain this case at $\d=0$.

Polynomial (\ref{13:13}) corresponds to equation (\ref{q2}). This
is the most simple case since the parameters are fixed already by
condition (\ref{idQ}).

In the case (\ref{4:4}), appropriate shifts and, if necessary,
changes of sign of the variables $x_2$, $x_3$, $x_4$ allow us to
bring $h^{12}$, $h^{23}$, $h^{34}$ into the form
$2h(x,y,\a)=\a^{-1}(x-y)^2-\d\a$ with $\d=1$. Analogously, in the
case (\ref{0:0}) appropriate M\"obius transforms of the general
form bring $h^{12}$, $h^{23}$, $h^{34}$ into the same form with
$\d=0$. In both cases, the invariants are $i_2=3\alpha^{-2}$,
$4i_3=\alpha^{-3}$, therefore the parameters of $h^{12}$ and
$h^{34}$ coincide and no further changes are necessary. The
resulting equation is (\ref{q1}).
\end{proof}


\section{ Classification of 3D-consistent systems of type $Q$}
\label{s:Q}

Theorem \ref{th:h} provides very strong necessary conditions for
3D-consistency in the case when all equations are of type $Q$. This
will allow us to classify such systems in this section. In this
final step we have to arrange the obtained equation around the cube
and to choose the parameters in such a way that the condition
(\ref{ABCABC}) is fulfilled. The effect of this condition may be a
change of sign or an inversion of one of the parameters.

In the following Theorem we return to the notation of the variables
and parameters according to the shifts on the lattice. The ordering
of the equations corresponds to the previous Theorem, and we name
these equations as in \cite{Adler_Bobenko_Suris_2003}.

\begin{theorem}\label{th:Qlist}
Any 3D-consistent system (\ref{6Q}) of type $Q$ is, up to M\"obius
transformations, one of the following list:
\begin{align}
\nonumber
 & \sn(\A{i})\sn(\A{j})\sn(\A{i}-\A{j})(k^2xx_ix_jx_{ij}+1)
  +\sn(\A{i})(xx_i+x_jx_{ij}) \\
\label{Q4}\tag{\hbox{$Q_4$}}
 &\qquad\qquad
  -\sn(\A{j})(xx_j+x_ix_{ij})-\sn(\A{i}-\A{j})(xx_{ij}+x_ix_j)=0,\\
\nonumber
 & \Bigl(\A{i}-\frac1{\A{i}}\Bigr)(xx_i+x_jx_{ij})
  -\Bigl(\A{j}-\frac1{\A{j}}\Bigr)(xx_j+x_ix_{ij})
  -\Bigl(\frac{\A{i}}{\A{j}}-\frac{\A{j}}{\A{i}}\Bigr)(xx_{ij}+x_ix_j)\\
\label{Q3}\tag{\hbox{$Q_3$}}
 & \qquad\qquad
  -\frac{\d}{4}\Bigl(\A{i}-\frac1{\A{i}}\Bigr)
   \Bigl(\A{j}-\frac1{\A{j}}\Bigr)
   \Bigl(\frac{\A{i}}{\A{j}}-\frac{\A{j}}{\A{i}}\Bigr)=0,\\
\nonumber
 & \A{i}(x-x_j)(x_i-x_{ij})-\A{j}(x-x_i)(x_j-x_{ij})
  +\A{i}\A{j}(\A{i}-\A{j})(x+x_i+x_j+x_{ij}) \\
\label{Q2}\tag{\hbox{$Q_2$}}
 & \qquad\qquad
  -\A{i}\A{j}(\A{i}-\A{j})((\A{i})^2-\A{i}\A{j}+(\A{j})^2)=0,\\
\label{Q1}\tag{\hbox{$Q_1$}}
 & \A{i}(x-x_j)(x_i-x_{ij})-\A{j}(x-x_i)(x_j-x_{ij})
 +\d\A{i}\A{j}(\A{i}-\A{j})=0.
\end{align}
\end{theorem}
\begin{proof}
First of all, note that equations of the different types
(\ref{q4})--(\ref{q1}) cannot be consistent with each other since
the corresponding singular curves are different. In particular,
the parameters $k^2$ in the case (\ref{q4}) and $\d$ in the cases
(\ref{q3}), (\ref{q1}) must be the same on each face of the cube.
Moreover, each equation of the list possesses the square symmetry,
that is, it is invariant with respect to the changes
$(x_1\leftrightarrow x_2$, $x_3\leftrightarrow x_4)$ and
$(x_1\leftrightarrow x_3$, $\a\leftrightarrow\b)$.

Therefore, the equations on all faces may differ only by the values
of $\a$ and $\b$. Consider the equations corresponding to three
faces meeting in one vertex, say $x$:
\[
 Q(x,x_1,x_2,x_{12},\a,\tilde\b)=0,\quad
 Q(x,x_2,x_3,x_{23},\b,\tilde\g)=0,\quad
 Q(x,x_3,x_1,x_{13},\g,\tilde\a)=0.
\]
Let
\[
 \d_{x_2,x_{12}}Q(x,x_1,x_2,x_{12},\a,\tilde\b)=\kappa(\a,\tilde\b)h(x,
x_1,\a).
\]
Then, due to the symmetry,
\[
 \d_{x_1,x_{12}}Q(x,x_1,x_2,x_{12},\a,\tilde\b)=\kappa(\tilde\b,\a)h(x,
x_2,\tilde\b)
\]
and, according to the Theorem \ref{th:h}, the parameters must be
related as follows:
\begin{gather*}
 \frac{h(x,x_1,\a)}{h(x,x_1,\tilde\a)}=m(\a,\tilde\a),\quad
 \frac{h(x,x_2,\b)}{h(x,x_2,\tilde\b)}=m(\b,\tilde\b),\quad
 \frac{h(x,x_3,\g)}{h(x,x_3,\tilde\g)}=m(\g,\tilde\g),\\
 \frac{\kappa(\a,\tilde\b)\kappa(\b,\tilde\g)\kappa(\g,\tilde\a)}
      {\kappa(\tilde\b,\a)\kappa(\tilde\g,\b)\kappa(\tilde\a,\g)}
  m(\a,\tilde\a)m(\b,\tilde\b)m(\g,\tilde\g)=-1.
\end{gather*}
In the case (\ref{q4}), a direct computation proves that
$\kappa(\a,\b)=2\sn(\a)\sn(\b)\sn(\a+\b)$ and
\[
 h(x,y,\a)=\frac{1}{2\sn(\a)}(k^2\sn^2(\a)x^2y^2+2\sn'(\a)xy-x^2-y^2+\sn^2(\a)),
\]
therefore $\tilde\a$ may take the values $\pm\a$ and analogously for
$\b,\g$. Obviously, up to a change of the numeration, two cases are
possible:
\[
 \tilde\a=-\a,\quad \tilde\b=-\b,\quad \tilde\g=-\g\qquad\text{or}\qquad
 \tilde\a=\a,\quad \tilde\b=\b,\quad \tilde\g=-\g.
\]
Moreover, this is actually only one case since we can make the
change $(\a,\tilde\b)\to(-\a,-\tilde\b)$ which keeps the equation
$Q(x,x_1,x_2,x_{12},\a,\tilde\b)=0$ invariant, as one can easily see
from (\ref{q4}). It is not difficult to check that we can always
adjust the signs on the whole cube as in the system (\ref{Q4}).

Consider now the case (\ref{q3}). Here
\begin{gather*}
 \kappa(\a,\b)=-\frac{(1-\a^2\b^2)(1-\a^2)(1-\b^2)}{\a^2\b^2},\\
 h(x,y,\a)=\frac{\a}{1-\a^2}(x^2+y^2)-\frac{1+\a^2}{1-\a^2}xy+\frac{(1-
\a^2)\d}{4\a}
\end{gather*}
and $\tilde\a=\a$ or $\tilde\a=1/\a$. Taking into account the
invariance of equation (\ref{q3}) with respect to the simultaneous
inversion of $\a,\b$, we can set, without loss of generality,
\[
 \tilde\a=1/\a,\quad \tilde\b=1/\b,\quad \tilde\g=1/\g
\]
resulting in the system (\ref{Q3}). In the cases (\ref{q2}),
(\ref{q1}) we have respectively
\begin{gather*}
 \kappa(\a,\b)=-4\a\b(\a+\b),\quad
 h(x,y,\a)=\frac{1}{4\a}(x-y)^2-\frac{\a}{2}(x+y)+\frac{\a^3}{4}, \\
 \kappa(\a,\b)=-2\a\b(\a+\b),\quad
 h(x,y,\a)=\frac{1}{2\a}(x-y)^2-\frac{\a\d}{2}
\end{gather*}
and we may set $\tilde\a=-\a$, $\tilde\b=-\b$, $\tilde\g=-\g$
exactly as before. This leads to the systems (\ref{Q2}), (\ref{Q1}).
\end{proof}

The master equation (\ref{Q4}) of the list was first derived in
\cite{A1} and further investigated in \cite{Adler_Suris_2004}. A Lax
representation for (\ref{Q4}) was found in \cite{N} with the help of
the method based on the three-dimensional consistency. Equations
(\ref{Q1}) and (\ref{Q3}$|_{\d=0}$) go back to \cite{QNCL}.
Equations (\ref{Q2}) and (\ref{Q3}$|_{\d=1}$) first appeared
explicitly in \cite{Adler_Bobenko_Suris_2003}.


\section{Examples of type $H$ systems}\label{s:H}

In contrast to the type $Q$ systems, the systems of type $H$ can be
considered as ``degenerate''. Their classification seems to be a
rather tedious task. Presently we cannot suggest any effective
procedure to solve this problem. On the other hand, the examples
given in the Section \ref{s:sing} demonstrate that this class should
not be just neglected as ``pathological''. Indeed, the discrete KdV
example (\ref{dKdV}) suggests that in some cases the degeneracy of
the biquadratics is just an unessential coincidence which does not
spoil the integrability properties of an equation. Here we consider
several more examples of this kind, corresponding to the cases
(\ref{22:22}), (\ref{22:22'}) at $\la\mu=0$, (\ref{4:4}) at $\la=0$,
and (\ref{0:0}) at $\ka\nu-\la\mu=0$ which were excluded in the
previous Section. It turns out that if we apply the same algorithm
in these cases (in spite of the fact that there is no justification
for this) then the list $H$ from our previous paper
\cite{Adler_Bobenko_Suris_2003} will be reproduced:
\begin{align}
\label{H3}\tag{\hbox{$H_3$}}
 & \A{i}(xx_i+x_jx_{ij})-\A{j}(xx_j+x_ix_{ij})+\d((\A{i})^2-(\A{j})^2)=0,
\\
\label{H2}\tag{\hbox{$H_2$}}
 & (x-x_{ij})(x_i-x_j)+(\A{j}-\A{i})(x+x_i+x_j+x_{ij})+(\A{j})^2-
(\A{i})^2=0,\\
\label{H1}\tag{\hbox{$H_1$}}
 & (x-x_{ij})(x_i-x_j)+\A{j}-\A{i}=0.
\end{align}
One may check directly that all statements of the Theorem
\ref{th:h} remain valid for these equations, in spite of the
degeneracy of the biquadratics.

Considering the asymmetric cases (\ref{112:22}), (\ref{13:4}),
(\ref{4:0}) with the different polynomials associated to the
different vertices, one finds that the following variants are
possible, up to the permutations. (There is clearly no distinction
between edges and diagonals when we are dealing with a single
equation.)
\begin{gather*}
 (x_1^2-1,x_2^2,x_3^2,x_4^2),\quad (x_1^2-1,x_2^2-1,x_3^2,x_4^2),\quad
 (x_1^2-1,x_2^2-1,x_3^2-1,x_4^2),\\
 (x_1,1,1,1),\quad (x_1,x_2,1,1),\quad (x_1,x_2,x_3,1), \\
 (1,0,0,0),\quad (1,1,0,0),\quad (1,1,1,0).
\end{gather*}
A direct check shows that the variants of type
 $\begin{pmatrix} r_2(x_4) & r_1(x_3) \\ r_1(x_1) & r_2(x_2) \end{pmatrix}$
are realizable and lead to the following list of 3D-consistent
equations:

\begin{align}
\label{epsH3}\tag{\hbox{$H^\eps_3$}}
  & \a(x_1x_2+x_3x_4)-\b(x_1x_4+x_2x_3)
  +(\a^2-\b^2)\Bigl(\d+\frac{\eps x_2x_4}{\a\b}\Bigr)=0,\\
\nonumber
  & (x_1-x_3)(x_2-x_4)+(\b-\a)(x_1+x_2+x_3+x_4)+\b^2-\a^2\\
\label{epsH2}\tag{\hbox{$H^\eps_2$}}
  & \qquad +\eps(\b-\a)(2x_2+\a+\b)(2x_4+\a+\b)+\eps(\b-\a)^3=0,\\
\label{epsH1}\tag{\hbox{$H^\eps_1$}}
  & (x_1-x_3)(x_2-x_4)+(\b-\a)(1+\eps x_2x_4)=0.
\end{align}
This list can be considered as a deformation of the list $H$ which
corresponds to the case $\eps=0$. However, we use the notation with
cyclic indices rather than shifts since due to the lack of symmetry
the arrangement of the equations on the faces of a cube requires a
more explicit description (see below). Note that in (\ref{epsH1})
the polynomial $1+\eps x_2x_4$ can be replaced by the polynomial
$\kappa x_2 x_4+\mu(x_2 + x_4)+\nu$ with arbitrary coefficients. The
corresponding biquadratic polynomials and their discriminants are
given in the following table (up to multiplication by a suitable
constant, $Q\to\mu(\a,\b)Q$):
\[
 \begin{array}{l|ccc}
  & h(x_1,x_2) & r_1(x_1) & r_2(x_2) \\
  \hline &&&\\[-2mm]
  (\ref{epsH3}) & x_1 x_2+\eps\a^{-1}x_2^2+\d\a & x_1^2-4\d\eps &
x_2^2\\[1mm]
  (\ref{epsH2}) & x_1+x_2+\a+ 2\eps(x_2+\a)^2  & 1-8\eps x_1   &
1 \\[1mm]
  (\ref{epsH1}) & 1+\eps x_2^2             & -4\eps      & 0 \\
 \end{array}
\]

Each of these equations possesses the rhombic symmetry
\[
 Q(x_1,x_2,x_3,x_4,\a,\b)=-Q(x_3,x_2,x_1,x_4,\b,\a)=-Q(x_1,x_4,x_3,
x_2,\b,\a),
\]
but not the square symmetry since the vertices $x_1,x_2$
correspond to polynomials with zeroes of different multiplicities.
The equation is 3D-consistent on the black-white lattice
$i+j+k\pmod2$. That is, each face must carry a copy of the
equation in such way that the parameters on opposite edges
coincide and the vertices $x,x_{12},x_{13},x_{23}$ are of the same
type (here we switch again to the notation where indices denote
shifts, as in Figure~\ref{f.cube}):
\[
 Q(x,x_i,x_{ij},x_j,\A{i},\A{j})=0,\quad
 Q(x_{ik},x_k,x_{jk},x_{123},\A{i},\A{j})=0,\quad \{i,j,k\}=\{1,2,3\}.
\]
Obviously, the equations on opposite faces of the cube do not
coincide, but the equation may nevertheless be extended to the
whole lattice $\Integer^3$. The tetrahedron property is fulfilled.

Finally, we notice that it is also possible to combine equations
with the square and trapezoidal symmetry. Consider equation
(\ref{Q1}) again. Let one pair of opposite faces carry the equations
\[
 Q_1(x,x_1,x_{12},x_2;\A1,\A2)_{\d=1}=0,\quad
 Q_1(x_3,x_{13},x_{123},x_{23};\A1,\A2)_{\d=0}=0,
\]
and let two other pairs carry the equations
\[
 Q(x,x_i,x_{i,3},x_3,\A{i},\eps)=0,\quad
 Q(x_j,x_{ij},x_{123},x_{j,3},\A{i},\eps)=0,\quad \{i,j\}=\{1,2\},
\]
where the polynomial
\[
 Q(x_1,x_2,x_3,x_4,\g,\eps)=(x_1-x_2)(x_3-x_4)+\g(\eps^{-1}-\eps x_3x_4)
\]
actually coincides with (\ref{epsH1}) up to the permutation of
$x_2, x_3$. This awkward structure is 3D-consistent and,
surprisingly, satisfies the tetrahedron property. It can be also
extended to the lattice $\Integer^3$.


\section{Concluding remarks}

Non-commutative analogues of some of the equations in the Q-list are
known. In particular, a quantum version of (\ref{Q1}$|_{\d=0}$)
appeared in \cite{Volkov}. In \cite{Bobenko_Suris_noncom} the
consistency approach was extended to the non-commutative context,
when the fields take values in an arbitrary associative algebra. The
definition of three-dimensional consistency remains the same in this
case, however, the assumption on the affine-linearity is replaced by
the requirement that the equation can be brought to the linear form
$px=q$ with respect to any variable $x$. These two properties are
not equivalent in the noncommutative case, as is seen from the
following examples. The first one was found in
\cite{Bobenko_Suris_noncom} and the other two by V.V.~Sokolov and
V.~Adler (unpublished):
\begin{align*}
\tag{\hbox{$\widehat Q_1|_{\d=0}$}}
  & \a^{(1)}(x-x_2)(x_2-x_{12})^{-1}=\a^{(2)}(x-x_1)(x_1-x_{12})^{-
1}, \\
  & \a^{(1)}(x_1-x_{12}+\a^{(2)})(x-x_1-\a^{(1)})^{-1} \\
\tag{\hbox{$\widehat Q_1|_{\d=1}$}}
  &\qquad\qquad =\a^{(2)}(x_2-x_{12}+\a^{(2)})(x-x_2-\a^{(2)})^{-1}, \\
  & (1-(\a^{(1)})^2)(x_1-\a^{(2)}x_{12})(\a^{(1)}x-x_1)^{-1}\\
\tag{\hbox{$\widehat Q_3|_{\d=0}$}}
  &\qquad\qquad =(1-(\a^{(2)})^2)(x_2-\a^{(1)}x_{12})(\a^{(2)}x-x_2)^{-
1}.
\end{align*}
The existence of non-commutative analogs of (\ref{Q2}),
(\ref{Q3}$|_{\d=1}$) and (\ref{Q4}) remains an open question.
Although the analysis of the singular solutions may still be useful
as a general principle, our technique is based on the algebraic
properties of affine-linear and biquadratic polynomials and is
therefore not applicable to this problem.

More general quantum systems with consistency property were found
recently in \cite{BazhSerg,BaMaSe}.


\end{document}